\documentclass{eptcs}
\providecommand{\event}{Linearity 2014} 

\usepackage{fixltx2e}
\usepackage{amsmath}
\usepackage{amssymb}
\usepackage{amsthm}
\usepackage{mathpartir2}
\usepackage{xspace}
\usepackage{graphicx}
\usepackage{fancyvrb}
\usepackage{listings}
\usepackage{comment}
\usepackage{booktabs}
\usepackage{hyperref}
\usepackage{plstx}
\usepackage{subfig}

\usepackage{float}
\restylefloat{figure}

\newtheorem{thm}{Theorem}
\newtheorem{lem}{Lemma}

\newcommand\V[1]{\texttt{\textit{#1}}}

\DefineVerbatimEnvironment
  {code}{Verbatim}
  {commandchars=\/\{\},
   codes={\catcode`$=3\relax},
   baselinestretch=.8,
   formatcom=\relax}

\fvset{commandchars=\/\{\},
       codes={\catcode`$=3\relax},
       formatcom=\relax}
\DefineShortVerb{\|}

\global\long\def\cinferlt#1#2#3{\inferrule*[lab=#1]{#2}{#3}}

\global\long\def\icr{\icr}

\global\long\def\fv{\mathop{\text{\rm fv}}}

\newcommand\textkw[1]{\textbf{#1}}
\newcommand\kw[2][\relax]{\textkw{#2$#1\relax$}\;}
\newcommand\ikw[2][\relax]{\;\textkw{#2$#1\relax$}\;}

\global\long\def\esub#1#2{\left\{  #1/#2\right\}  }

\global\long\def\ein{\ikw{in}}

\global\long\def\eletp#1#2#3{\kw{letp}\left(#1,#2\right)=#3}

\global\long\def\eapp#1#2{\mathop{#1}#2}

\global\long\def\elam#1#2{\mathbf{\lambda}#1.\,#2}

\global\long\def\eqlam#1#2#3{\mathbf{\lambda}^{#1}#2.\,#3}

\global\long\def\etlam#1#2#3{\Lambda#1\left[#2\right].\,#3}

\global\long\def\etapp#1#2{#1\left[#2\right]}

\global\long\def\epair#1#2{\left(#1,#2\right)}

\global\long\def\einl#1{\kw{Inl}#1}

\global\long\def\einr#1{\kw{Inr}#1}

\global\long\def\ematch#1#2#3#4#5{\kw{case}#1\ikw{of}\einl{#2}\rightarrow#3;\,\einr{#4}\rightarrow#5}

\global\long\def\eqnew#1#2{\kw[^{#1}]{new}#2}

\global\long\def\eqget#1#2{\kw[^{#1}]{release}#2}

\global\long\def\eqswp#1#2#3{\kw[^{#1}]{swap}#2\ikw{with}#3}

\global\long\def\elab{\ell}

\global\long\def\edup#1#2#3{\kw{dup}#1\ikw{as}#2,#3}

\global\long\def\edrop#1{\kw{drop}#1}

\global\long\def\eunit{()}

\global\long\def\eaq{\ensuremath{\mathit{aq}}}
\global\long\def\erq{\ensuremath{\mathit{rq}}}
\global\long\def\tW{\textsf{w}}
\global\long\def\tS{\textsf{s}}

\global\long\def\ehole{\left[\cdot\right]}

\global\long\def\tcdup{\mathop{\textsf{\upshape Dup}}}

\global\long\def\tcdrop{\mathop{\textsf{\upshape Drop}}}

\global\long\def\tck{K{}}

\global\long\def\tcentail{\Vdash}

\global\long\def\tcenv{\Gamma^{\text{\rm is}}}

\global\long\def\sccompat{\mathrel{\smile}}

\global\long\def\tmap{\mkern.5mu{:}\mkern.5mu}

\global\long\def\ltqmap#1#2{\mathrel{\mapsto}_{#1}^{#2}}

\global\long\def\ltwmap#1{\ltqmap {\tW}{#1}}

\global\long\def\ltsmap{\mathrel{\mapsto}_{\tS}}

\global\long\def\tje#1#2#3#4#5{#1;#2;#3\vdash#4\mathrel:#5}

\global\long\def\tqarrow#1#2#3{#1\overset{#3}{\longrightarrow}#2}

\global\long\def\tU{\textsf{U}}

\global\long\def\tR{\textsf{R}}

\global\long\def\tA{\textsf{A}}

\global\long\def\tL{\textsf{L}}

\global\long\def\tprod#1#2{#1\times#2}

\global\long\def\tsum#1#2{#1+#2}

\global\long\def\tqref#1#2{\textsf{Ref}^{\text{\,#1}}\,#2}

\global\long\def\tsref#1{\tqref {\tS}{#1}}

\global\long\def\twref#1{\tqref {\tW}{#1}}

\global\long\def\tunit{\textsf{Unit}}

\global\long\def\tsub#1#2{\left\{  #1/#2\right\}  }

\global\long\def\dom{\mathop{\text{\rm dom}}}

\global\long\def\tsch#1#2#3{\forall#1.#2\Rightarrow#3}

\global\long\def\preds{P}

\global\long\def\alphais{\overline{\alpha_{i}}}

\global\long\def\tjes#1#2#3#4{\tje{\preds}{#1}{#2}{#3}{#4}}

\global\long\def\nmstate#1#2{(#1\ ;\ #2)}

\global\long\def\nmtrans{\ \longmapsto\ }

\global\long\def\lqmap#1{\mathrel{\mapsto}^{#1}}

\global\long\def\tjs#1#2#3{#1\vdash_{\!\!s}#2\mathrel:#3}

\global\long\def\tjc#1#2#3{\vdash_{\!\!c}\nmstate{#1}{#2}\mathrel:#3}

\global\long\def\floc{\mathop{\text{\rm locs}}}

\global\long\def\nmincr#1#2{\mathop{\text{\rm incr}}(#1;#2)}

\global\long\def\nmdecr#1#2{\mathop{\text{\rm decr}}(#1;#2)}

\global\long\def\cll{\hat\lambda_{\mathit{lin}}}

\global\long\def\clam#1#2{\mathbf{\lambda}#1.#2}

\global\long\def\cpair#1#2{#1 \otimes #2}

\global\long\def\cchoice#1#2{#1 \mathbin{\&} #2}

\global\long\def\cdup#1#2{\mathbf{dup}\ #1\ \mathbf{in}\ #2}

\global\long\def\cdrop#1#2{\mathbf{drop}\ #1\ \mathbf{in}\ #2}

\global\long\def\evp{+}

\global\long\def\evsub{\sqsubseteq}

\global\long\def\cwf#1#2{#1\vdash#2}

\global\long\def\erase{\mathop{\text{\rm erase}}}

\global\long\def\inferf{\mathop{\text{\rm infer}}}

\global\long\def\alms{^{a}\lambda_{ms}}

\global\long\def\lcl{\lambda_{\mathit{cl}}}

\global\long\def\lrural{\lambda^{\textsf{refURAL}}}

\global\long\def\Constrain{{\text{\rm Constrain}}}

\newcommand\rulename[1]{\textsc{#1}}

\newcommand\noopsort[1]{\relax}
\renewcommand{\MathparLineskip}{\lineskiplimit=.6\baselineskip\lineskip=.6\baselineskip plus .2\baselineskip}

\title{Type Classes for Lightweight Substructural Types}
\def\titlerunning{Type Classes for Lightweight Substructural Types}

\author{
  Edward Gan
  \institute{Facebook, Menlo Park}
  \email{edgan8@gmail.com}
  \and
  Jesse A.~Tov
  \institute{Northeastern University, Boston}
  \email{tov@ccs.neu.edu}
  \and
  Greg Morrisett
  \institute{Harvard University, Cambridge}
  \email{greg@eecs.harvard.edu}
}
\def\authorrunning{Gan, Tov \& Morrisett}

\begin{document}

\maketitle

\begin{abstract} 

Linear and substructural types are powerful tools, but adding them to
standard functional programming languages often means introducing extra
annotations and typing machinery. We propose a lightweight substructural
type system design that recasts the structural rules of weakening and
contraction as type classes; we demonstrate this design in a prototype
language, Clamp.

Clamp supports polymorphic substructural types as well as an expressive
system of mutable references. At the same time, it adds little
additional overhead
to a standard Damas--Hindley--Milner
type system enriched with type classes. We have
established type safety for the core model and implemented a type
checker with type inference in Haskell.

\end{abstract}

\section{Introduction} 

Type classes \cite{qualtypes,wadleradhoc} provide a way to constrain
types by the operations they support. If the type class
predicate $\tcdup \alpha$ indicates when assumptions of type $\alpha$ are
subject to contraction (duplication), and $\tcdrop \alpha$ indicates
whether they are
subject to weakening (dropping), then linear, relevant, affine, and
unlimited typing disciplines are all enforced by
some subset of these classes. 
Linear types, then, are types that satisfy neither $\tcdup$ nor $\tcdrop$.
This idea, suggested in one author's
dissertation~\cite{tovthesis}, forms the basis of
our prototype substructural programming language Clamp.

Clamp programs are written in a Haskell-like external language in which
weakening and contraction are implicit. This is easier for programmers to work
with, but to specify the type system and semantics the external language is
elaborated into an internal language that is linear (i.e. variables are used
exactly once.) The internal language provides explicit \emph{dup} and
\emph{drop} operations, which impose the corresponding type class constraints
on their arguments.
Thus, in the internal language one might think of
\emph{dup} and
\emph{drop} as functions with these qualified types:
\begin{align*}
\textit{dup} :&\; \tsch{\alpha}{\tcdup \alpha}
  {\alpha \rightarrow \tprod{\alpha}{\alpha}}\\
\textit{drop} :&\; \tsch{\alpha\beta}{\tcdrop \alpha}
  {\alpha \rightarrow \beta \rightarrow \beta}
\end{align*}

In the internal language all nonlinear usage is mediated by
the {dup} and {drop} operations. For example,
the internal language term $\lambda x .\, x+x$ is ill formed
because it uses variable $x$ twice, but the term
$$
\lambda x .\,\textbf{let } (x_1, x_2) = \textit{dup } x \textbf{ in } x_1+x_2
$$
is well typed. Because elaboration into
the internal language ensures that the resulting
program is linear, it can then be checked using nearly-standard
Damas--Hindley--Milner type reconstruction \cite{damas82} with type
classes~\cite{qualtypes,wadleradhoc};
improper duplication and dropping is indicated by unsatisfiable type class
constraints.

\paragraph{Contributions.}
We believe that Clamp offers substructural types with less
fuss than many
prior approaches to programmer-facing substructural type systems.
Throughout the design, we leverage standard type class machinery
to deal with most of the constraints imposed by substructural types.
Implementing type inference for Clamp (\S\ref{sec:implementing})
is straightforward and it is also easy to extend the system with custom
resource aware structures (\S\ref{sec:instances}).
The specific contributions in this paper include:

\begin{itemize}

\item a type system design with polymorphic substructural types and a
type safety theorem~(\S\ref{sec:formalizing});

\item a flexible system for managing weak and strong references
(\S\ref{sec:instances});

\item a type checker with type inference derived from a type checker for
Haskell (\S\ref{sec:implementing}); and

\item a {dup}-and-{drop}--insertion algorithm that is in some
sense optimal (\S\ref{sec:inferring}).

\end{itemize} 

\begin{figure}
\begin{code}
/V{fst}    :: Drop /V{b} => (/V{a}, /V{b}) -U> /V{a}
/V{fst}     = /K\(/V{x}, /V{y}) -U> /V{x}

/V{constU} :: (Dup /V{a}, Drop /V{a}, Drop /V{b}) => /V{a} -U> /V{b} -U> /V{a}
/V{constU}  = /K\/V{x} -U> /K\/V{y} -U> /V{x}

/V{constL} :: Drop /V{b} => /V{a} -U> /V{b} -L> /V{a}
/V{constL}  = /K\/V{x} -U> \/V{y} -L> /V{x}
\end{code}
\caption{Prelude functions with inferred signatures}
\label{fig:prelude}
\end{figure}

\subsection{Clamp Basics \label{sec:overview}}

In this section, we introduce the Clamp external language, in which
{dup} and {drop} operations are implicit. The concrete syntax
is borrowed from Haskell, but one prominent difference in Clamp
is that each function type and term must be annotated
with one of four substructural qualifiers: |U| for unlimited, |R|
for relevant, |A| for affine, or |L| for linear.

Three examples of Clamp functions, translated from the Haskell standard
prelude, appear in figure~\ref{fig:prelude}. Their types need not be
written explicitly, and are inferred by Clamp's type checker.

Consider the |/V{fst}| function, which projects the first component of
a pair. Because we would like be able to use
library functions any number of times or not at all, we annotate the
arrow in the lambda expression with qualifier~|U|. This annotation
determines the function type's structural properties---meaning, in this case,
that \V{fst} satisfies both $\tcdup$ and $\tcdrop$. (Note that this is a property of the
function \emph{itself}, not of how it treats its argument.) Because \V{fst}
does not use the second component of the pair, this induces the
|Drop /V{b}| constraint on type variable \V{b}. In particular,
elaboration into the internal language inserts a \V{drop} operation
for \V{y} to make the term linear:
  |/K\(/V{x}, /V{y}) -U> /V{drop} /V{y} /V{x}|.
  The presence of \V{drop}, which disposes of its first argument
and returns its second, causes the |Drop| type class constraint to be
inferred.

Function \V{constU} imposes a similar constraint on its second
argument, but it also requires the type of its first argument be unlimited. This is
because \V{constU} returns an unlimited closure containing the first
argument in its environment. The argument is effectively duplicated or
discarded along with the closure, so it inherits the same structural
restrictions. Alternatively, we can lift this restriction with
\V{constL}, which returns a linear closure and thus allows the first
argument to be linear.


\section{Formalizing $\lcl$ \label{sec:formalizing}}

To validate the soundness of our approach, we have developed~$\lcl$, a core
model of the Clamp internal language. $\lcl$ is based on System~F~\cite{girardF}
with a few modifications:
variable bindings are treated linearly, arrows are annotated with
qualifiers, and type class
constraints~\cite{qualtypes,wadleradhoc} are added under universal quantifiers. 
As an example of how one can define custom usage-aware datatypes in Clamp,
$\lcl$ also includes a variety of operations for working with mutable
references. 

The $\lcl$~type system shares many similarities with Tov's
core Alms calculus: $\alms$~\cite{practical-affine}.
Unlike the external Clamp language prototype~(\S\ref{sec:implementing}),
$\lcl$ provides first-class polymorphism and does not support type
inference.

\begin{figure}
\begin{plstx}
  (terms): e
  ::= x
    | v
    | \eapp{e_{1}}{e_{2}}
    | \etapp e{\overline{\tau_{i}}}
    | \epair{e_{1}}{e_{2}}
    | \eletp{x_{1}}{x_{2}}e\ein e_{2}
    | \einl e
    | \einr e
    | \ematch e{x_{1}}{e_{1}}{x_{2}}{e_{2}}
    | \eqnew{\erq}e
    | \eqget{\erq}e
    | \eqswp{\erq}{e_{1}}{e_{2}}
    | \edup{e_{1}}{x_{1}}{x_{2}\ein e_{2}}
    | \edrop{e_{1}}\ein e_{2} \\
  (values): v
  ::=\eqlam{\eaq}{x{:}\tau}e
    | \etlam{\alphais}{\preds}v
    | \epair{v_{1}}{v_{2}}
    | \einl v
    | \einr v
    | \elab
    | \eunit \\
  (types): \tau
  ::=\alpha
    | \tqarrow{\tau_{1}}{\tau_{2}}{\eaq}
    | \tprod{\tau_{1}}{\tau_{2}}
    | \tsum{\tau_{1}}{\tau_{2}}
    | \tunit
    | \tqref{\erq}{\tau}
    | \tsch{\overline{\alpha_{i}}}{\preds}{\tau}
    \\[4pt]
  (constraints):
  \preds ::= (\tck_1 \tau_1, \dots, \tck_n \tau_n) \\
  (reference qualifiers):
  \erq ::=\tS\;\textit{(strong)}
    | \tW\;\textit{(weak)}\\
  (arrow qualifiers):
  \eaq ::=\tU\;\textit{(unlimited)}
    | \tR\;\textit{(relevant)}
    | \tA\;\textit{(affine)}
    | \tL\;\textit{(linear)}\\
  (predicate constructors):
  K ::=\tcdup | \tcdrop \\
\end{plstx}
\caption{Syntax of $\lcl$ \label{fig:lcl-Syntax}}
\end{figure}

\subsection{Syntax of $\lcl$}

The syntax of~$\lcl$ appears in figure~\ref{fig:lcl-Syntax}. Most of the
language is standard, but notably arrow types and  $\lambda$~terms in Clamp
are annotated with an \emph{arrow qualifier} ($\eaq$). These annotations
determine which structural operations a function supports, as well as the
corresponding constraints imposed on the types in its closure environment.
Unlike some presentations of linear logic, $\tqarrow{}{}{\eaq}$  here
constrains usage of the function itself, not usage of the function's argument.
Thus one can call {dup} on a $\tqarrow{}{}{\tU}$ arrow but not on an
$\tqarrow{}{}{\tL}$ arrow. Type abstractions specify the type class constraints
that they abstract over; their bodies are restricted to values, so unlike
$\lambda$ terms, type abstractions do not need an arrow qualifier.

The $\eqnew{\erq}e$ and $\eqget{\erq}e$ forms introduce and eliminate mutable
references. Each comes in two flavors depending on its  \emph{reference
qualifier} ($\erq$), which records whether the reference supports strong or
merely weak updates. Weak (conventional) updates must preserve a reference
cell's type, but strong updates can modify both the value and type of a cell. 

Form $\eqswp{\erq}{e_{1}}{e_{2}}$ provides linear access
to a reference by exchanging its contents for a different value. The
$\mathbf{release}$ operator deallocates a cell and returns its contents if it
is not aliased. Store locations ($\elab$) appear at run time but are not
written by the programmer.

To incorporate type classes, universal types may include constraints on their type
variables. A constraint $\preds$ denotes a set of atomic predicate constraints
$\tck \tau$, each of which is a predicate constructor $\tck$ applied to a type.
For the sake of our current analysis, $K$ is either $\tcdup$ or $\tcdrop$.

\subsection{Semantics}

\begin{figure}
\begin{plstx}
  (evaluation contexts):
  E ::=\ehole
    | \eapp Ee
    | \eapp vE
    | \etapp E{\overline{\tau_{i}}}
    | \epair Ee
    | \epair vE
    | \einl E\mid\einr E
    | \ematch E{x_{1}}{e_{1}}{x_{2}}{e_{2}}
    | \eletp{x_{1}}{x_{2}}E\ein e
    | \eqnew{\erq}E\mid\eqget{\erq}E
    | \eqswp{\erq}Ee\mid\eqswp{\erq}vE
    | \edup E{x_{1}}{x_{2}}\ein e\mid\edrop E\ein e\\
  (stores):
  \mu ::=\elab\lqmap iv,\mu | \cdot \\
\end{plstx}

\caption{Runtime structures\label{fig:Runtime-Structures}}
\end{figure}

\begin{figure}
\[
\begin{aligned}
  \nmstate{\mu}{\eapp{\left(\elam{x{:}\tau}e\right)}v} 
    & \nmtrans\nmstate{\mu}{\esub vxe}\\
  \nmstate{\mu}{\etapp{\left(\etlam{\alphais}{\preds}v\right)}
    {\overline{\tau_{i}}}} 
    & \nmtrans\nmstate{\mu}{\overline{\tsub{\tau_{i}}{\alpha_{i}}}v} 
    \\
  \nmstate{\mu}{\eqnew{\erq}{v}}
    & \nmtrans\nmstate{\mu,\elab\lqmap {1}v}{\elab}
    \qquad \text{$\elab$ fresh}
    \\
  \nmstate{\mu,\elab\lqmap iv_{1}}{\eqswp{\erq}{\elab}{v_{2}}} 
    & \nmtrans\nmstate{\mu,\elab\lqmap iv_{2}}{\epair{\elab}{v_{1}}} 
    \\
  \nmstate{\mu,\elab\lqmap 1v}{\eqget \tW{\elab}} 
    & \nmtrans\nmstate{\mu}{\einl v} 
    \\
  \nmstate{\mu,\elab\lqmap{i}v}{\eqget \tW{\elab}} 
    & \nmtrans\nmstate{\mu,\elab\lqmap{i-1}v}{\einr{\eunit}} 
    \qquad \text{when $i>1$}
    \\
  \nmstate{\mu,\elab\lqmap 1v}{\eqget \tS{\elab}} 
    & \nmtrans\nmstate{\mu}v
    \\
  \nmstate{\mu}{\edup v{x_{1}}{x_{2}}\ein e}
    & \nmtrans\nmstate{\nmincr{\floc\left(v\right)}{\mu}}
    {\esub v{x_{2}}\esub v{x_{1}}e} 
    \\
  \nmstate{\mu}{\edrop v\ein e} 
    & \nmtrans\nmstate{\nmdecr{\floc\left(v\right)}{\mu}}e 
    \\
  \nmstate{\mu_{1}}{E\left[e_{1}\right]} 
    & \nmtrans\nmstate{\mu_{2}}{E\left[e_{2}\right]}
    \qquad\text{when }
    \nmstate{\mu_{1}}{e_{1}}\nmtrans\nmstate{\mu_{2}}{e_{2}}
\end{aligned}
\]

\caption{Small-step relation \label{fig:Small-Step-Relation}}
\end{figure}

The execution of $\lcl$ terms can be defined by a call-by-value small-step
semantics with evaluation contexts and a global, reference-counted store
$\mu$. The run-time structures needed to define this small step relation are
given in figure~\ref{fig:Runtime-Structures}. Reference counts are used to
track when reference cells can be safely deallocated in the presence of
aliasing. Exchange properties for the store are implicitly assumed. A
selection of the small step relation rules are given in 
figure~\ref{fig:Small-Step-Relation}, focusing on the rules for reference 
cells and substructural operations. Most of the complexity here comes 
from the reference counts.

The $\mathbf{swap}$ operator exchanges the contents of a cell in the
heap with a different value.  The $\mathbf{release}$ operator deallocates a
cell and returns its contents if it is not aliased. However, if a cell
has been aliased it decrements the reference count and returns a unit. The dup
and drop operators manipulate reference counts as expected.

A few metafunctions given in figure~\ref{fig:refcount-Manage} are necessary to
maintain reference counts, similar to those in \cite{computational-linear}. The
functions $\nmincr{\elab}{\mu}$ and $\nmdecr{\elab}{\mu}$ allow us
to increment and decrement reference counts in the heap. Incrementing a
location is straightforward, but a decrement must be defined recursively since
deallocating the last pointer to a reference cell involves decrementing the
reference counts of all cells the deallocated contents originally pointed to.

The $\floc$ meta-function is a convenient way of extracting the multiset
of locations that a value uses. Note that the $+$ (or $\uplus$) operator
is a multiset operator which additively combines occurrences, and is used
again in section~\ref{sec:core}.  The $\floc$ function is also
designed to operate on well-typed terms, so it only looks at one branch
of a $\mathbf{case}$ expression, assuming that the other branch must
share the same location typing context.

\begin{figure}
\centering
\begin{minipage}{0.45\linewidth}
\[
\begin{aligned}
  \nmincr{\elab}{\elab\lqmap jv,\mu} 
    & =\elab\lqmap{j+1}v,\mu \\
  \nmdecr{\elab}{\elab\lqmap jv,\mu} 
    & =
      \begin{cases}
        \nmdecr{\floc\left(v\right)}{\mu}
            & \text{$j = 1$} \\
        \elab\lqmap{j-1}v,\mu
            & \text{$j > 1$}
      \end{cases} \\[4pt]
  \nmincr{\{\elab_1,\ldots,\elab_k\}}{\mu}
    & =\nmincr{\elab_1}{\cdots\nmincr{\elab_k}{\mu}} \\
  \nmdecr{\{\elab_1,\ldots,\elab_k\}}{\mu}
    & =\nmdecr{\elab_1}{\cdots\nmdecr{\elab_k}{\mu}}
\end{aligned}
\]
\end{minipage}
\qquad
\begin{minipage}{0.45\linewidth}
\[
\begin{aligned}
  \floc\left(\elab\right) & =\left\{ \elab\right\} \\
  \floc\left(\eqlam{\eaq}{x{:}\tau}e\right) 
    & =\floc\left(e\right) \\
  \floc\left(e_{1}\ e_{2}\right) 
    & =\floc\left(e_{1}\right)+\floc\left(e_{2}\right)\\
  \floc\left(
  \begin{array}{@{}l@{}}
        \mathbf{case\ }e\mathbf{\ of \ } \\
        \einl{x_1} \to e_1; \\
        \einr{x_2} \to e_2
  \end{array}
  \right)
    & =\floc\left(e\right)+\floc\left(e_{1}\right)\\
  \cdots
\end{aligned}
\]
\vspace{-1.5em}
\end{minipage}
\caption{Reference count management \label{fig:refcount-Manage}}

\end{figure}

\subsection{Term Typing}

Variable contexts $\Gamma::=x_{1}{:}\tau_{1},\dots,x_{n}{:}\tau_{n}$ 
associate variables with types,
where each variable appears at most once.
Location contexts (store typings) 
$\Sigma::=\elab_{s}\ltsmap\tau_{s},\dots,\elab_{w}\ltwmap{k_{w}}\tau_{w},\dots$ 
associate locations
($\elab$) with their reference types $\tqref{}{\tau}$, 
and distinguish between strong and weak
locations; weak locations carry a reference count $k$ to
track aliasing.

Linearity is enforced in $\lcl$ via standard context-splitting.
Because $\Gamma$ and $\Sigma$ are linear environments, we need operations to join them.
The join operation~$+$ is defined only on pairs of
compatible environments, written $\Gamma_{1}\sccompat\Gamma_{2}$ and
$\Sigma_{1}\sccompat\Sigma_{2}$. Two variable contexts are compatible so long
as they are disjoint. Two location
contexts are compatible if the strong locations are disjoint and the weak
locations in their intersection agree on their types. Joining variable contexts
appends the two sets of bindings together, while joining location contexts
also involves adding the reference counts of any shared weak locations.
Contexts are identified up to permutation.

\begin{figure}
\begin{mathpar}
\cinferlt{Var}{\ }{\tje{\preds}{x\tmap\tau}{\cdot}x{\tau}}
\and
\cinferlt{TAbs}
  {\tje{\preds_{1},\preds_{2}}
    {\Gamma}{\Sigma}v{\tau}\icr
    \dom\preds_{2}\subseteq\alphais}
  {\tje{\preds_{1}}{\Gamma}{\Sigma}
    {\etlam{\alphais}{\preds_{2}}v}{\tsch{\alphais}{\preds_{2}}{\tau}}}
\and
\cinferlt{TApp}
  {\tje{\preds_{1}}{\Gamma}{\Sigma}e{\tsch{\alphais}{\preds_{2}}{\tau}}
    \icr\preds_{1}\tcentail\overline{\tsub{\tau_{i}}{\alpha_{i}}}\preds_{2}
    }
  {\tje{\preds_{1}}{\Gamma}{\Sigma}{\etapp e{\overline{\tau_{i}}}}
    {\overline{\tsub{\tau_{i}}{\alpha_{i}}}\tau}}
\and
\cinferlt{Abs}
  {\tje{\preds}{\Gamma,x\tmap\tau_{1}}{\Sigma}e{\tau_{2}}\icr
    \preds\tcentail\Constrain^\eaq(\Gamma;\Sigma)}
  {\tje{\preds}{\Gamma}{\Sigma}{\eqlam{\eaq}{x\:{:}\:\tau_{1}}e}
    {\tqarrow{\tau_{1}}{\tau_{2}}{\eaq}}}
\and
\cinferlt{App}
  {\tje{\preds}{\Gamma_{1}}{\Sigma_{1}}{e_{1}}
    {\tqarrow{\tau_{2}}{\tau}{\eaq}}
    \icr\tje{\preds}{\Gamma_{2}}{\Sigma_{2}}{e_{2}}{\tau_{2}}}
  {\tje{\preds}{\Gamma_{1}+\Gamma_{2}}{\Sigma_{1}+\Sigma_{2}}
    {\eapp{e_{1}}{e_{2}}}{\tau}}
\and
\cinferlt{Case}
  {\tjes{\Gamma_{1}}{\Sigma_{1}}{e_{1}}{\tsum{\tau_{11}}
    {\tau_{12}}}
    \icr\tjes{\Gamma_{2},x_{21}\tmap\tau_{11}}
    {\Sigma_{2}}{e_{21}}{\tau_{2}}
    \icr\tjes{\Gamma_{2},x_{22}\tmap\tau_{12}}
    {\Sigma_{2}}{e_{22}}{\tau_{2}}}
  {\tje{\preds}{\Gamma_{1}+\Gamma_{2}}{\Sigma_{1}+\Sigma_{2}}
    {\ematch{e_{1}}{x_{21}}{e_{21}}{x_{22}}{e_{22}}}{\tau_{2}}}
\and
\cinferlt{Dup}
  {\tje{\preds}{\Gamma_{1}}{\Sigma_{1}}{e_{1}}{\tau_{1}}
    \icr\icr\tje{\preds}{\Gamma_{2},x_{1}\tmap\tau_{1},x_{2}\tmap\tau_{1}}
    {\Sigma_{2}}{e_{2}}{\tau_{2}}
    \icr\preds\tcentail\tcdup\tau_{1}}
  {\tje{\preds}{\Gamma_{1}+\Gamma_{2}}{\Sigma_{1}+\Sigma_{2}}
    {\edup{e_{1}}{x_{1}}{x_{2}}\ein e_{2}}{\tau_{2}}}
\and
\cinferlt{Drop}
  {\tje{\preds}{\Gamma_{1}}{\Sigma_{1}}{e_{1}}{\tau_{1}}
    \icr\icr\tje{\preds}{\Gamma_{2}}{\Sigma_{2}}{e_{2}}{\tau_{2}}
    \icr\preds\tcentail\tcdrop\tau_{1}}
  {\tje{\preds}{\Gamma_{1}+\Gamma_{2}}{\Sigma_{1}+\Sigma_{2}}
    {\edrop{e_{1}}\ein e_{2}}{\tau_{2}}}
\end{mathpar}
\vspace{-1em}
\caption{Selected $\lcl$ term typing rules \label{fig:lcl-Expression-Typing}}
\end{figure}

\begin{figure}
\begin{mathpar}
\and
\cinferlt{New}
  {\tje{\preds}{\Gamma}{\Sigma}e{\tau}}
  {\tje{\preds}{\Gamma}{\Sigma}{\eqnew{\erq}e}{\tqref{\erq}{\tau}}}
\and
\cinferlt{LocW}{\ }
  {\tje{\preds}{\cdot}{\elab\ltwmap 1\tau}{\elab}{\twref{\tau}}}
\and
\cinferlt{LocS}{\ }
  {\tje{\preds}{\cdot}{\elab\ltsmap\tau}{\elab}{\tsref{\tau}}}
\and
\cinferlt{ReleaseW}
  {\tje{\preds}{\Gamma}{\Sigma}e{\tqref{\erq}{\tau}}}
  {\tje{\preds}{\Gamma}{\Sigma}{\eqget we}{\tsum{\tunit}{\tau}}}
\and
\cinferlt{ReleaseS}
  {\tje{\preds}{\Gamma}{\Sigma}e{\tqref \tS{\tau}}}
  {\tje{\preds}{\Gamma}{\Sigma}{\eqget se}{\tau}}
\and
\cinferlt{SwapW}
  {\tje{\preds}{\Gamma_{1}}{\Sigma_{1}}{e_{1}}{\tqref{\erq}{\tau}}
    \icr\tje{\preds}{\Gamma_{2}}{\Sigma_{2}}{e_{2}}{\tau}}
  {\tje{\preds}{\Gamma_{1}+\Gamma_{2}}{\Sigma_{1}+\Sigma_{2}}
    {\eqswp \tW{e_{1}}{e_{2}}}{\tprod{\tqref{\erq}{\tau}}{\tau}}}
\and
\cinferlt{SwapS}
  {\tje{\preds}{\Gamma_{1}}{\Sigma_{1}}{e_{1}}{\tqref \tS{\tau_{1}}}
    \icr\tje{\preds}{\Gamma_{2}}{\Sigma_{2}}{e_{2}}{\tau_{2}}}
  {\tje{\preds}{\Gamma_{1}+\Gamma_{2}}{\Sigma_{1}+\Sigma_{2}}
    {\eqswp \tS{e_{1}}{e_{2}}}{\tprod{\tqref \tS{\tau_{2}}}{\tau_{1}}}}
\end{mathpar}
\vspace{-1em}
\caption{$\lcl$ term reference cell typing rules \label{fig:lcl-Reference-Typing}}
\end{figure}

The term typing judgment ($\tje{\preds}{\Gamma}{\Sigma}e{\tau}$) assigns
term~$e$ type~$\tau$ under constraint, variable, and location contexts
$\preds$, $\Gamma$, and $\Sigma$. Selected typing rules for the core
language appear in figure~\ref{fig:lcl-Expression-Typing}, and the typing
rules for reference cells are given in figure~\ref{fig:lcl-Reference-Typing}.
Consistency conditions $\Sigma_{1}\sccompat\Sigma_{2}$ and
$\Gamma_{1}\sccompat\Gamma_{2}$ are assumed whenever contexts are combined.
The core language typing rules split and share the linear contexts as needed,
but are otherwise a natural extension of System~F to support type class
constraints.

We impose a syntactic restriction, similar to Haskell~98's context reduction
restrictions~\cite{haskell98}, on the form of constraints in type
schemes introduced by the \rulename{TAbs} rule: type abstractions may only
constrain the type variables that they bind, and not compound or unrelated
types. This simplifies induction over typing derivations for \rulename{TAbs}
since it means that no constraints on external type variables
can be introduced by a type abstraction.
 Additionally, in rule \rulename{Abs}, the variable and location
contexts are constrained by the function's arrow qualifier, to ensure
that values captured by the closure support any structural operations
that might be applied to the closure itself; this constraint must be
entailed ($\tcentail$) by the constraint context. Here $\Constrain^\eaq(\Gamma;\Sigma)$
is shorthand for the appropriate set of Dup and Drop constraints
applied to every type mapped in $\Gamma$ and $\Sigma$, so that
for instance
$\Constrain^\tL$ imposes no constraints, while
$\Constrain^\tR$ imposes only $\tcdup$ constraints.

The $\textkw{dup}$ and $\textkw{drop}$ forms constrain the types of their parameters
in the expected way, by requiring their types to be members of the
$\tcdup$ or $\tcdrop$ type classes, respectively (again entailed by the
constraint context).

Since $\lcl$ supports both strong and weak references with different
substructural properties, there are a variety of typing rules governing their
usage. The \textbf{swap} operation needs to return both an updated reference
and the old contents, so it packages those in a pair. Weak and strong forms of
reference cell operations are provided, and it is safe to apply the weak
operations to both strong and weak references.  The $\eqget{\erq}e$ forms are
used to deallocate a reference cell and possibly retrieve its contents.
Notably, in the case of a weak reference, since the contents could be linear,
we preserve its linearity while allowing for aliasing by returning the
contents of the reference only when the last alias to the cell is released,
and unit otherwise.

\subsection{Type Class Instances\label{sec:instances}}

Throughout the type system, type class constraints are propagated via
entailment, $\preds_1 \tcentail \preds_2$, which specifies when one set of
type class predicates ($\preds_2$) is implied by another ($\preds_1$) in the
context of the fixed background instance environment~$\tcenv$. For example,
entailment allows our type system to derive that $\tprod{\tunit}{\tunit}$ is
duplicable because $\tunit$ is. Rules for entailment are given by
Jones~\cite{qualtypes} and adapt naturally to this setting. The substructural
essence of the type class system in Clamp is the set of base $\tcdup$ and
$\tcdrop$ instances $\tcenv$, which appears in figure~\ref{fig:instances}.

\begin{figure}
\begin{align*}
  (\tcdup \alpha_1,\tcdup \alpha_2) & 
    \Rightarrow\tcdup\,(\tprod {\alpha_1} {\alpha_2} ) &
  (\tcdrop \alpha_1,\tcdrop \alpha_2) & 
    \Rightarrow\tcdrop\,(\tprod {\alpha_1} {\alpha_2} )\\
  (\tcdup \alpha_1,\tcdup \alpha_2) & 
    \Rightarrow\tcdup\,(\tsum {\alpha_1} {\alpha_2} ) & 
  (\tcdrop \alpha_1,\tcdrop \alpha_2) 
    & \Rightarrow\tcdrop\,(\tsum {\alpha_1} {\alpha_2} )
\end{align*}
\vspace{-2em}
\begin{align*}
  () & \Rightarrow\tcdup\,(\tqarrow {\alpha_1} {\alpha_2} {\tU}) & 
  () & \Rightarrow\tcdrop\,(\tqarrow {\alpha_1} {\alpha_2} {\tU}) &
  () & \Rightarrow\tcdup\tunit & () & \Rightarrow\tcdrop\tunit \\
  () & \Rightarrow\tcdup\,(\tqarrow {\alpha_1} {\alpha_2} {\tR}) & 
  () & \Rightarrow\tcdrop\,(\tqarrow {\alpha_1} {\alpha_2} {\tA}) &
  () & \Rightarrow\tcdup\,(\twref \alpha ) & 
  (\tcdrop \alpha) & \Rightarrow\tcdrop(\tqref{\erq}\alpha )
\end{align*}
\vspace{-1em}
\caption{$\tcdup$ and $\tcdrop$ instances
  \label{fig:instances}}
\end{figure}

Since pairs and sums contain values
that might be copied or ignored along with the pair or sum value, their
instance rules require instances for their components. Functions impose
constraints on their closure environments when they are assigned a qualifier
during term typing, so the instance rules for arrows depend only on the arrow
qualifier.

Dealing correctly with references is more subtle, as seen in
$\lrural$~\cite{stepindexed}. In Clamp, some references support strong
updates, which can change not only the value but the type of a mutable
reference. However it is unsafe to alias a reference cell whose type might
change.

In $\lrural$, the restrictions on reference types are given in a sizable table,
but $\tcdup$ and $\tcdrop$ instances make it easy to express these restrictions in
Clamp. Clamp classifies references by the kind of updates they
support: strong or weak. This is specified by the $\erq$ qualifier 
in the $\tqref{\erq}$ type.

Qualitatively, the constraints we impose are that:

\begin{itemize}

\item Strong references may not be duplicated.

\item Only references with droppable contents may be dropped.

\item Only strong references support direct deallocation. 

\item Weak references can be deallocated, but only return their contents when unaliased.
  
\end{itemize}

The above four rules capture the same restrictions as $\lrural$ references.
They also increase the expressiveness of the system by explicitly
distinguishing weak and strong references and allowing for the deallocation of
weak references. They are expressed in $\lcl$ with two type class instances
and the typing judgments \rulename{ReleaseW} and \rulename{ReleaseS}.

As an example of the kinds of structures we can build using these rules, consider
the type 
$$\twref{\left(\text{fhandle}\right)}$$ 
for a linear file handle $\text{fhandle}$. This
weak reference can be aliased to provide shared access to the file handle, but
cannot be dropped based on the type class instances, since $\text{fhandle}$ cannot
be dropped. Anyone that uses this reference must release
the reference and close the file if necessary.

\subsection{Type Safety}

Here we sketch part of the type safety proof; more details may be found in
Gan's thesis~\cite{ganthesis}.

The bulk of the work goes into proving preservation, and the key lemma in
proving preservation relates constraints to bindings. Intuitively, this lemma
says that structural constraints on a value's type respect the structural
constraints of everything the value contains or points to, via the variable
and location contexts. Syntactic forms like $\tcdup\Gamma$ are used to denote
the set of $\tcdup$ constraints on all types in $\Gamma$, and similarly
$\tcdup\Sigma$ applies $\tcdup$ to all of the $\tqref{rq}{\tau}$ types mapped
by $\Sigma$.  Note that the lemma does not hold for arbitrary expressions.

\begin{lem}[Constraints capture bindings]
\label{lem:capture}
Suppose that $\tje{\preds}{\Gamma}{\Sigma}v{\tau}$. 
If $\preds\tcentail\tcdup\tau$
then \mbox{$\preds\tcentail(\tcdup\Sigma,\tcdup\Gamma)$};
if $\preds\tcentail\tcdrop\tau$
then $\preds\tcentail(\tcdrop\Sigma,\tcdrop\Gamma)$.
\end{lem}

\begin{proof}
By induction on the typing derivation for $v$.
\end{proof}

Lemma~\ref{lem:capture} is essential to proving the substitution lemma 
(Lemma \ref{lem:subst}).
\begin{lem}[Substitution] \label{lem:subst}

If
\begin{itemize}
\item $\tje{\preds}{\Gamma,x\tmap\tau_{x}}{\Sigma_{1}}e{\tau}$\,,
\item $\tje{\preds}{\cdot}{\Sigma_{2}}v{\tau_{x}}$\,, and
\item $\Sigma_{1}\sccompat\Sigma_{2}$,
\end{itemize}
then $\tje{\preds}{\Gamma}{\Sigma_{1}+\Sigma_{2}}
        {\esub vxe}{\tau}$
\end{lem}

\begin{proof}
By induction on the typing derivation for $e$, making use of 
lemma~\ref{lem:capture} in the $\lambda$ case.
\end{proof}

In proving Preservation, it is also useful to separate out a Replacement
Lemma which specifies exactly how substitution interacts with evaluation
contexts.

\begin{lem}[Replacement]
\label{Replacement:}

If $\tje{\preds}{\Gamma}{\Sigma}{E\left[M\right]}{\tau}$ then $\exists\tau',\Sigma_{1},\Sigma_{2},\Gamma_{1},\Gamma_{2}$
such that 
\begin{itemize}
\item $\Sigma=\Sigma_{1}+\Sigma_{2}$ and $\Gamma=\Gamma_{1}+\Gamma_{2}$
and $\tje{\preds}{\Gamma_{1}}{\Sigma_{1}}M{\tau'}$ and furthermore
\item If $\tje{\preds}{\Gamma_{1}^{'}}{\Sigma_{1}^{'}}{M'}{\tau'}$ with
$\Gamma_{1}^{'}\sccompat\Gamma_{2}$ and $\Sigma_{1}^{'}\sccompat\Sigma_{2}$,
then $\tje{\preds}{\Gamma_{1}^{'}+\Gamma_{2}}{\Sigma_{1}^{'}+\Sigma_{2}}{E\left[M'\right]}{\tau}$
for any $M',\Gamma_{1}^{'},\Sigma_{1}^{'}$ 
\end{itemize}
\end{lem}
\begin{proof}
By induction on $E$.
\end{proof}

Another key lemma for proving preservation relates the $\floc$ function
used to maintain dynamic reference counts with the store context that
a value requires. To state this lemma, we overload the $\floc$ function
to also return the multiset of occurrences (multiple for reference counted weak 
location stores) of locations in the domain of a store context.

\begin{lem}[Store Contexts map Free Locations]
\label{lem:floc-ctxt}

If $\tje{\preds}{\Gamma}{\Sigma}e{\tau}$ then $\floc\left(e\right)=\floc\left(\Sigma\right)$.
\end{lem}

Finally, to prove preservation and type soundness we need to introduce store
and configuration typings which are given in Figure~\ref{fig:Store-Typing}.
The remainder of the type soundness proof is then mostly standard.

\begin{figure}
\begin{mathpar}
\cinferlt{St-Nil}{\ }{\tjs{\Sigma}{\cdot}{\cdot}}
\and
\cinferlt{St-ConsW}
  {\tjs{\Sigma_{1}}{\mu}{\Sigma_{2}
    \icr\tje{\cdot}{\cdot}{\Sigma_{v}}v{\tau}}}
  {\tjs{\Sigma_{1}+\Sigma_{v}}
    {\mu,\elab\lqmap iv}{\Sigma_{2},\elab\ltqmap wi\tau}}
\and
\cinferlt{St-ConsS}
  {\tjs{\Sigma_{1}}{\mu}{\Sigma_{2}\icr
    \tje{\cdot}{\cdot}{\Sigma_{v}}v{\tau}}}
  {\tjs{\Sigma_{1}+\Sigma_{v}}{\mu,\elab\lqmap 1v}
    {\Sigma_{2},\elab\ltsmap\tau}}
\and
\cinferlt{Conf}
  {\tjs{\Sigma_{1}}{\mu}{\Sigma_{1}+\Sigma_{2}\icr
    \tje{\cdot}{\cdot}{\Sigma_{2}}e{\tau}}}
  {\tjc{\mu}e{\tau}}

\end{mathpar}
\vspace{-1em}
\caption{Store and configuration typing \label{fig:Store-Typing}}

\end{figure}

\begin{lem}[Preservation]
\label{Preservation:}

If $\tjc{\mu_{1}}{e_{1}}{\tau}$ and $\nmstate{\mu_{1}}{e_{1}}\nmtrans\nmstate{\mu_{2}}{e_{2}}$
then $\tjc{\mu_{2}}{e_{2}}{\tau}$
\end{lem}

\begin{thm}[Type soundness]
If $\tjc{\cdot}e{\tau}$ then either it diverges or it reduces to
a value configuration $\nmstate{\mu}v$ such that $\tjc{\mu}v{\tau}$.
\end{thm}


\section{Implementing the Clamp Type Checker \label{sec:implementing}}

We have implemented a type checker that infers Damas--Hindley--Milner style
type schemes for
Clamp terms. The type checker is an extension of Jones's ``Typing Haskell
in Haskell'' type checker~\cite{typinghaskell}.
Its source code may be found at
\url{https://github.com/edgan8/clampcheck}.
 
The process of modifying a Haskell type checker to support Clamp was
straightforward and illustrates one of the strengths of Clamp's design: It
requires only small and orthogonal additions to a language like Haskell. Besides
adding qualifiers to arrow types, we made three additions
to a Haskell type checker:
\begin{enumerate}
\item an elaboration pass that inserts dups and drops,
\item $\tcdup$ and $\tcdrop$ type classes and instances, and
\item substructural qualifiers and constraints on arrow types.
\end{enumerate}

\subsection{Inferring dups and drops. \label{sec:inferring}}

The elaboration pass is the bridge between a concise user-facing
language and leveraging conventional, nonlinear type checking techniques. The
pass takes as input a term with arbitrary variable
usages; it inserts the appropriate dup and drop operations and renames
the duplicated copies so that in the resulting term all variable usage is
strictly linear. Structural properties are then enforced by the
constraints imposed by dup and drop.

Since different elaborations can lead to different static and dynamic
semantics, we have proven that our algorithm generates an \emph{optimal}
elaboration in two senses:
\begin{itemize}
\item It minimizes the program's live variables.
\item It imposes minimal type class constraints.
\end{itemize}
In what follows we define a core linear language to formalize
and estabilsh these two points.

\paragraph{An Abstract Linear Language\label{sec:core}}

To focus on the essential problems, we can work with an
\emph{abstraction} of the linear $\lambda$~calculus, $\cll$. By modeling only
usage and binding, $\cll$ allows us to focus on inserting dup and drop
operations independently of the particular types and term forms of a language.
Its syntax is given in figure~\ref{fig:lclsyntax}. Extending the results to
cover other term forms is straightforward.

\begin{figure}
\begin{plstx}
  (unannotated terms):
  e ::=x
  | \clam x e
  | \cpair{e_{1}}{e_{2}}
  | \cchoice{e_{1}}{e_{2}}\\
  (annotated terms):
  ae ::=x
  | \clam x ae
  | \cpair{ae_{1}}{ae_{2}}
  | \cchoice{ae_{1}}{ae_{2}}
  | \cdup{\Gamma}{ae}
  | \cdrop{\Gamma}{ae} \\
\end{plstx}
\caption{$\cll$ syntax \label{fig:lclsyntax}}
\end{figure}

The product expression $\cpair{e_{1}}{e_{2}}$ abstracts multiplicative
forms such as pairs and function applications---that is,
pairs of expressions where both will be evaluated.
The sum expression $\cchoice{e_{1}}{e_{2}}$ abstracts additive forms
such as linear logic's additive conjuction,
and the relationship between the branches of a $\textbf{case}$---that is,
pairs of expressions where exactly one will be evaluated.

Expressions, $e$, are unannotated and don't explicitly satisfy linear usage
constraints. Annotated expression, $ae$, use dup and drop
operations to explicitly specify nonlinear usage of variables. The dup and
drop operations work over contexts, $\Gamma$, which are multisets of variables

The contexts $\Gamma$ in $\cll$ manage scope and binding by restricting
contraction and weakening to explicit dup and drop annotations, but do not
track the types of variables. In order to avoid the messy but straightforward
process of generating names and renaming variables when inserting a dup, we
think of contexts as multisets (\emph{e.g.,} $\{ x, x, y, z, \ldots\}$) of 
in-scope variables, or equivalently as functions from variables to natural
numbers. Thus $\Gamma(x)$ will be used to denote the number of times $x$ appears
in $\Gamma$.

\begin{figure}
\begin{mathpar}
\cinferlt{L-Var}{\ }{\cwf{\left\{ x\right\} }x} \and
\cinferlt{L-Abs}{\cwf{\Gamma+\left\{ x\right\} }{ae} \icr x\notin\Gamma} 
  {\cwf{\Gamma}{\clam x{ae}}} \and
\cinferlt{L-Pair}{\cwf{\Gamma_{1}}{ae_{1}}\icr\cwf{\Gamma_{2}}{ae_{2}}}
  {\cwf{\Gamma_{1}\evp\Gamma_{2}}{\cpair{ae_{1}}{ae_{2}}}} \and
\cinferlt{L-Choice}{\cwf{\Gamma}{ae_{1}}\icr\cwf{\Gamma}{ae_{2}}}
  {\cwf{\Gamma}{\cchoice{ae_{1}}{ae_{2}}}} \and
\cinferlt{L-Dup}{\cwf{\Gamma_{1}\evp\Gamma_{2}\evp\Gamma_{2}}{ae}}
  {\cwf{\Gamma_{1}\evp\Gamma_{2}}{\cdup{\Gamma_{2}}{ae}}} \and
\cinferlt{L-Drop}{\cwf{\Gamma_{1}}{ae}}
  {\cwf{\Gamma_{1}\evp\Gamma_{2}}{\cdrop{\Gamma_{2}}{ae}}}
\end{mathpar}
\vspace{-1em}
\caption{Annotated expression well-formedness \label{fig:Well-Formedness-Rules}}
\end{figure}

We use these multiset contexts to define a notion
of \emph{well-formedness} in figure~\ref{fig:Well-Formedness-Rules},
which describes when an annotated term $ae$ in $\cll$ properly accounts for
all nonlinear usage of its variables through explicit dup and drop
operations.

\paragraph{Inference Algorithm}

An inference algorithm for annotating terms is given in 
figure~\ref{fig:Inference-Algorithm}.  
The strategy is to recursively transform a term
bottom-up based on the free variables $\fv$ in each recursively
transformed sub-term. Note that the $\fv$ function always returns a \emph{set} 
and that $\cap$ and $\setminus$ denote the standard set intersection and
difference operators.

Dup operations are inserted where the free variables of
two sub-terms of a multiplicative form (\emph{e.g.,} application, but not
branching) are discovered to intersect; drops are added under binders when the
bound variable is not free in its scope, and when a variable used in one
branch (say, of an if-then-else) is not free in the other. 

\begin{figure}
\begin{align*}
\inferf\left(x\right) & = x \\
\inferf\left(\clam xe\right) & =
    \begin{cases}
    \clam x{\inferf\left(e\right)}
    & \quad\textnormal{if }x\in\fv(e);\\
    \clam x{\cdrop x{\inferf\left(e\right)}}
    & \quad\textnormal{otherwise}
    \end{cases}\\
\inferf\left(\cpair{e_1}{e_2}\right) &
= \cdup{\fv(e_1) \cap \fv(e_2)}
  {\cpair{\inferf(e_1)}{\inferf(e_2)}} \\
\inferf\left(\cchoice {e_1}{e_2} \right) &
= \cchoice{\left(\cdrop{\fv(e_2)\setminus\fv(e_1)}{\inferf(e_1)}\right)}
  {\left(\cdrop{\fv(e_1)\setminus\fv(e_2)}{\inferf(e_2)}\right)}
\end{align*}
\vspace{-1em}
\caption{Inference algorithm \label{fig:Inference-Algorithm}}
\end{figure}

We outline the key steps in our
optimality argument here; additional details may be found in
\cite{ganthesis}. Lemma~\ref{lem:soundfv} states that the algorithm
is sound.

\begin{lem}[Soundness]
$\cwf{\fv(e)}{\inferf(e)}$
\label{lem:soundfv}
\end{lem}

If $ae$ is an annotation of $e$, then $\fv(e) \evsub \fv(ae)$, so
lemma~\ref{lem:minctx} shows that our algorithm in fact generates an
annotation which requires a minimal context $\Gamma$ for well-formedness.

\begin{lem}[Minimal contexts]
If $\cwf{\Gamma}{ae}$ then $\fv(ae) \evsub \Gamma$.
\label{lem:minctx}
\end{lem}

With some technical lemmas, we can then prove that the algorithm introduces no
unnecessary dups or drops on variables. In order to compare different
potential annotations of the same term, we define a function $\erase :
ae \to e$, which removes dup and drop annotations from an annotated
term in the straightforward way, yielding an unannotated term. It should be
evident that $\erase(\inferf(e)) = e$.

\begin{lem}[Forced drop]
  If $\cwf{\Gamma}{ae}$, $\Gamma\left(x\right)\geq1$, and 
  $x\notin\fv\left(\erase\left(ae\right)\right)$, 
  then $ae$ contains a subterm $\cdrop{\Gamma'}{ae'}$ 
  such that $x\in\Gamma'$.
\label{lem:fdrop}
\end{lem}

\begin{lem}[Forced dup]
  If $\cwf{\Gamma}{ae}$, $\Gamma\left(x\right)\leq1$, and there
  exists a subderivation $\cwf{\Gamma_{s}}{ae_{s}}$ of
  $\cwf{\Gamma}{ae}$ with $\Gamma_{s}\left(x\right)\geq2$, 
  then $ae$ contains a subterm $\cdup{\Gamma'}{ae'}$ 
  such that $x\in\Gamma'$
\label{lem:fdup}
\end{lem}

\begin{lem}[No Unnecessary Drops]
Let $ae=\inferf\left(e\right)$. If ae contains
a subterm $\cdrop{\Gamma_{d}}{ae_{s}}$ with $x\in\Gamma_{d}$ then
any other well-formed $ae'$ with $\erase\left(ae'\right)=e$ contains
a subterm $\cdrop{\Gamma_{d}^{'}}{ae_{s}^{'}}$ with $x\in\Gamma_{d}^{'}$.
\label{lem:unnecdrop}\end{lem}

\begin{lem}[No Unnecessary Dups]
Let $ae=\inferf\left(e\right)$. If ae contains
a subterm $\cdup{\Gamma_{d}}{ae_{s}}$ with $x\in\Gamma_{d}$ then
any other $ae'$ with $\erase\left(ae'\right)=e$ and $\cwf{\Gamma^{'}}{ae'}$
for $\Gamma^{'}\left(x\right)\leq1$ contains a subterm $\cdup{\Gamma_{d}^{'}}{ae_{s}^{'}}$
with $x\in\Gamma_{d}^{'}$.\label{lem:unnecdup}\end{lem}

\subsection{Constraint processing.}

After the {dup} and {drop} insertion pass, type
inference can proceed without needing to count variable
usages or split contexts, since the insertion pass has made
every {dup} and {drop} operation explicit. With the exception
of the extra constraints imposed on closure environments, 
inferring types for the Clamp internal language is like
inferring types for Haskell. For the constraint solver,
the type classes $\tcdup$ and $\tcdrop$ and their instances
are no different than any other type class.

Type checking in this system is thus separated into two self-contained steps:
first, usage analysis as performed by elaboration, and second,
checking substructural constraints in the
same manner as any other type class system.
A similar division was used
by de~Vries et al.\ to integrate a uniqueness typing system into
Damas--Hindley--Milner~\cite{uniqueness-simple}.

In Table~\ref{tab:Type-Checker-Code} we present the sizes of the components of
our Clamp implementation. It compares favorably to the implementation of
languages such as Alms \cite{practical-affine}, whose type inference engine is
15,000 lines of Haskell. The dup/drop insertion sits on top of the stack and
is the main addition we have had to make to a Haskell type checker design.
Besides that, we have included the base set of type class instances and
altered arrow kinds throughout.

\begin{table}
\begin{center}
\begin{tabular}{l r}
Component & Lines of code\tabularnewline
\midrule
Dup/drop insertion & 160\tabularnewline
Type class instances & 60\tabularnewline
\midrule
Syntax and Types & 703\tabularnewline
Parser/Lexer & 319\tabularnewline
Unification Engine & 373\tabularnewline
\end{tabular}
\end{center}

\caption{Type checker code breakdown\label{tab:Type-Checker-Code}}
\end{table}

\section{Related Work}

Of the existing work in linear type systems, we will focus here on those which
develop general purpose polymorphic linear types. Research on the mathematical
expressiveness of linearity \cite{recursor} or more tailored use cases
\cite{sessiongay} for instance often do not aim at broad usability and
polymorphism.

The first linear type systems derive directly from intuitionistic linear
logic, and use the exponential~``!'' to indicate types that support structural
operations \cite{computationabramsky,ill}. Some later type systems, in order
to support parametric polymorphism over linearity, replace ``!'' with types
composed of a qualifier and a pretype \cite{stepindexed,atapl}, so that all
types in these languages have a form like~$^{q}\overline\tau$. Similarly, the
Clean programming language makes use of qualifier variables and inequalities
to capture a range of substructural polymorphism 
\cite{uniqueness, uniqueness-simple}. 
Though Clean uses uniqueness rather than linear types, many of its
design decisions can be applied in a linear settings as well.

More recent languages such as Alms \cite{practical-affine} and $F^{\circ}$
\cite{fpop} eliminate the notational overhead of annotating every type with
substructural qualifiers by using distinct kinds to separate substructural
types. Thus, rather than working with types like $^{\tA}\textsf{file}$, a
\textsf{file} type in Alms can be defined to have kind $\tA$. Like Clean, Alms
is highly polymorphic, but it makes use of compound qualifier expressions on
function types, as well as dependent kinds and sub-kinding.

Compared to type systems like those of Clean and Alms, we believe Clamp offers
advantages in simplicity and extensibility. Like Alms and $F^{\circ}$, Clamp
avoids the burden in Clean of annotating every  type with a qualifier. Type
classes themselves are a general and powerful feature; for a language that is
going to have type classes anyway, the Clamp approach allows adding the full
spectrum of URAL types with little additional complexity for programmers and
type checkers. Programmers already familiar with type classes will be well
prepared to understand Clamp-style substructural types.

Further, type classes provide a clean formalism for constraining state-aware
datatypes such as the system of weak and strong mutable references found in
Clamp (\S\ref{sec:instances}). Finally, we anticipate that user-defined |Dup|
and |Drop| instances, not yet supported by Clamp, will allow defining custom
destructors and copy constructors, which should enable a variety of resource
management strategies.

However, compared to Alms and Clean, Clamp does not provide as much
polymorphism because each arrow is assigned a concrete qualifier.
Consider, for instance, a \V{curry} function in Clamp.
Unlike in Alms or Clean, Clamp requires different versions for different
desired structural properties.
For instance, two possible type schemes for a \V{curry} function are
\begin{code}
    (Dup /V{a}, Drop /V{a}) => ((/V{a}, /V{b}) -U> /V{c}) -U> /V{a} -U> /V{b} -U> /V{c}
\end{code}
and
\begin{code}
    ((/V{a}, /V{b}) -L> /V{c}) -U> /V{a} -L> /V{b} -L> c/rm.
\end{code}
We believe that
extending Clamp with qualifier variables and type class
implications could increase its expressiveness to the point where
\V{curry} has a principal typing.

\section{Future Work\label{sec:future}}

\subsection{Custom \lstinline!dup! and \lstinline!drop!}

In Clamp's current design, the semantics of dup and drop are fixed. Allowing
programmers instead to define their own implementations of dup and drop on
user-defined types would enable scenarios similar to those possible in C++ via
copy constructors and destructors.  Programmers could then define data types
that automatically manage resources in ways that meet particular needs; but
unlike in C++, we believe this could be done without unsafe operations.

Programmers could define types and instances to manage their memory in
whatever way is most appropriate, for instance by choosing between deep-copy
and shallow-copy dup operations (or perhaps some hybrid approach), or between
eager and lazy drop operations. For such a system to be practical, it is
important that the dup/drop-insertion algorithm be easy to understand, since
the insertion of dups and drops can affect the dynamic semantics of the
language.

\subsection{Polymorphic Arrows}

In most cases, Clamp allows programmers to define functions that
are inherently polymorphic over the substructural properties of their
arguments. In these cases, a function that uses one of its arguments
linearly can accept any type, whether it satisfies Dup or Drop, for
that argument. However, this is not the case for function types,
which are annotated with fixed qualifiers determined at the function
definition point.

Alms is able to accommodate more polymorphic arrow types by introducing a
subtyping relation on qualified arrows \cite{practical-affine}, and it refines
the types of arrows further by introducing usage qualifiers that depend on
the substructural properties of type variables. This allows one to write a
function whose qualifier is inferred from the closure environment and inherits
any polymorphism present in that environment.

There are many options for increasing the polymorphic expressiveness of
Clamp without resorting to the complexities of subtyping. One could make
qualifiers into first class types and allow quantification over
qualifiers in arrow types. This is implemented in an ad-hoc way in our
current type checker. 

The idea of expanding the language of qualifiers could also be
profitable in Clamp. This often means annotating arrow types with the types
of their closure environments, and in a sense assigning them a
closure-converted type. The dup and drop instances for arrows could then
use the closure environment types to determine the arrow type's
substructural properties.

\subsection{Implementation}

The current implementation of the Clamp type checker reflects $\lcl$, and
could benefit from the addition of some standard language features found in
full Haskell. In particular, the addition of algebraic datatypes, user-defined
instance rules, and a module system would allow programmers to define
libraries that expose custom types with varying substructural properties.

It would also be interesting to implement a compiler for Clamp that
takes advantage of substructural properties for reference-counted memory
management \cite{computational-linear}. Eagerly reusing the
storage occupied by linear and affine values may offer particular performance
advantages, and
substructural analyses can enable a variety of other optimizations
as well \cite{linquick, onceuponpolymorphic}.

\section{Conclusion}

Clamp introduces techniques that make it easier and more desirable to add
substructural types to functional programming languages. The external /
internal language distinction gives us both a programmer friendly syntax and a
direct path to type inference. The $\tcdup$ and $\tcdrop$ classes
also support polymorphism over the URAL lattice and can represent state aware types such
as strong and weak references. Type classes are an expressive and 
well-established language feature, and Clamp shows that they can serve as a
base for substructural types.





\bibliographystyle{eptcs}

\bibliography{Linear}
\end{document}